\documentclass[11pt]{article}
\usepackage{cite}
\usepackage{epstopdf}
\usepackage{graphicx}
\usepackage{subfigure}
\usepackage{geometry}
\usepackage[affil-it]{authblk}
\usepackage{amsmath}
\usepackage{amsthm}
\usepackage{amssymb}
\usepackage{url}
\usepackage{thmtools}
\usepackage[linesnumbered,ruled,vlined]{algorithm2e}
\usepackage{mathtools}
\usepackage{color}
\usepackage{fixmath}
\usepackage{array}
\usepackage{multirow}
\usepackage{etoolbox}
\usepackage{authblk}

\newtheorem{claim}{Claim}

{\bf}{\it}
{\bf}{\it}
{\bf}{\it}

\begin{document}
\title{No Need for Recovery:\\ A Simple Two-Step Byzantine Consensus}
\author[1]{Tung-Wei Kuo\thanks{twkuo@cs.nccu.edu.tw}}
\author[2]{Kung Chen\thanks{chenk@nccu.edu.tw}}
\affil[1]{Department of Computer Science, National Chengchi University, Taiwan}
\affil[2]{Department of Management Information Systems, National Chengchi University, Taiwan}
\date{\vspace{-5ex}}
\maketitle

\begin{abstract} \normalsize
In this paper, we give a deterministic two-step Byzantine consensus protocol that 
achieves safety and liveness. A two-step Byzantine consensus protocol only needs two 
communication steps to commit in the absence of faults.
Most two-step Byzantine consensus protocols exploit optimism
and require a recovery protocol in the presence of faults. 
In this paper, we give a simple two-step Byzantine consensus protocol 
that does not need a recovery protocol. 
\end{abstract}
\section{Introduction}
We consider the Byzantine agreement problem. 
Let $n$ and $f$ be the number of nodes (e.g., processors or replicas)
and the number of faulty nodes, respectively.
In this problem, each node has an initial value,
and nodes exchange messages to reach an agreement. 
Specifically, we need to design a message exchange protocol (or consensus algorithm) 
so that after the protocol terminates, all the non-faulty nodes output (or commit) 
the same value, and this value is the initial value of some node.
In other words, the consensus algorithm must guarantee safety.
Moreover, the protocol must terminate eventually, i.e., guarantee liveness.
In this paper, we consider the partially synchronous model.
Specifically, let $D(t)$ be the transmission delay of a message sent at time $t$. 
In the partially synchronous model, $D(t)$ does not grow faster than $t$ indefinitely.

Our goal is to design a two-step consensus algorithm. 
In one step, a node can 1) send messages, 2) receive messages, 
and 3) do local computation, in that order~\cite{Bosco}. 
A consensus algorithm is two-step if all non-faulty nodes 
can commit after two steps in the absence of faults.
It has been shown that to solve the Byzantine agreement problem 
by a two-step consensus algorithm, $n \geq 5f+1$ must hold~\cite{FaB}. 
Thus, in this paper, we assume $n = 5f+1$.

Several two-step consensus algorithms have been proposed 
to solve the Byzantine agreement problem~\cite{FaB, Zyz, SBFT}. 
These solutions proceed in rounds, and a round consists of two steps in normal operation. 
However, it has been pointed out that 
FaB~\cite{FaB} and Zyzzyva~\cite{Zyz} cannot guarantee 
both safety and liveness~\cite{revisit}.
Moreover, these consensus algorithms exploit optimism in their design 
and invoke additional recovery protocols 
when normal operation fails~\cite{FaB, Zyz, SBFT}.
Thus, these solutions may need more than two steps in a round 
due to faulty behavior or long communication delay.

In this paper, we give a simple two-step consensus algorithm 
without the use of any recovery protocol. 
In our solution, a round always consists of only two steps: 
1) a leader, which is chosen in a round-robin fashion, broadcasts a proposal, 
and 2) all nodes vote and collect votes.
Our solution and analysis are inspired by MSig-BFT, 
which is a three-step protocol~\cite{MSig-BFT}.
Like MSig-BFT, the leader may not be allowed to broadcast a proposal 
if the network is in bad condition. 
An interesting property of our solution is that nodes may reach consensus in a round, 
even if the leader chosen in that round is faulty or suffers from long transmission delay.
Such a property can thus mitigate the harm caused by faulty nodes and transmission delay. 

\section{The Two-Step Consensus Algorithm}
\label{sec:algo}
\noindent \textbf{The first step: propose}. 
At the beginning of round $r$, the leader sends a Proposal message, 
which contains a candidate value $b$ and the current round $r$, to all nodes.
We will describe this step in detail after the next step is introduced.
For a Proposal message $p$, $p.R$ and $p.B$ denote the round 
in which $p$ is generated and the candidate value contained in $p$.
The pseudocode of the first step is given in Algorithm~\ref{algo1}.
\begin{algorithm}[t]\label{algo1}
\caption{Propose: from the viewpoint of node $u$}
    \KwIn{The initial value $b_{in}$ of $u$ and $lockset(r-1)$}
    $ldr \leftarrow$ The leader of round $r$\\    
    \If{$u = ldr$}{
        Construct a Proposal message $p$\\ 
        $p.R \leftarrow r$\\    
        \tcc{Determine $p.B$ and send $p$ to all nodes}     
        \If{$r = 1$}{
            $p.B \leftarrow b_{in}$ and send $p$ to all nodes\\
        }\ElseIf{$|lockset(r-1)| \geq 4f+1$}{
            \If{$\exists \text{ \upshape a value } b \text{ \upshape such that }
                 b \neq \varnothing \text{ \upshape and } 
                |\{v|v \in lockset(r-1), v.B = b\}| \geq 2f+1$}{
                Let $b$ be any value satisfying the above constraint\\                 
                $p.B \leftarrow b$ and send $p$ to all nodes\\
            }\Else{
                $p.B \leftarrow b_{in}$ and send $p$ to all nodes\\
            }             
        }    	
    }
\end{algorithm}

\noindent \textbf{The second step: vote}.
Once a node $u$ receives a valid Proposal message $p$, 
$u$ then broadcasts a Vote message containing candidate value $p.B$.
Note that a node broadcasts at most one Vote message in a round.
If $u$ receives $4f+1$ Vote messages before a predetermined timeout $TO_{commit}$ expires, 
and these $4f+1$ Vote messages contain the same non-empty candidate value $b$, 
$u$ then commits $b$. On the other hand, if $u$ cannot commit before $TO_{commit}$ 
expires, $u$ then goes to the next round, and $u$ needs to store the candidate value 
for which it votes. Specifically, let $b'$ be the candidate value contained in 
the Vote message broadcast by $u$ in round $r$.
In round $r+1$, if $u$ cannot receive a valid Proposal message 
before another predetermined timeout $TO_{vote}$ $(TO_{vote} < TO_{commit})$ expires, 
$u$ then broadcasts a Vote message containing $b'$. 
Note that in the first round (i.e., $r = 1$), 
if $u$ cannot receive a valid Proposal message before $TO_{vote}$ expires, 
$u$ broadcasts a Vote message containing an empty candidate value $\varnothing$. 
To achieve liveness in the partially synchronous model, 
whenever a node goes to the next round, the lengths of the two timeouts are doubled.
For a Vote message $v$, we use $v.R$ and $v.B$ to denote the round 
in which $v$ is generated and the candidate value contained in $v$.
We summarize this step from the viewpoint of node $u$ in Algorithm~\ref{algo3}.

\begin{algorithm}[t]\label{algo3}
\caption{Vote: from the viewpoint of node $u$}
    Construct a Vote message $v$, and set $v.R \leftarrow r$\\
    \tcc{Determine $v.B$ and broadcast $v$}
    \If{$u \text{ \upshape receives a valid Proposal message } p \text{ \upshape 
        of round } r$ $\text{\upshape before } TO_{vote} \text { \upshape expires}$}{
        $v.B \leftarrow p.B$ and broadcast $v$\\      	
    }\Else{
        \If{$r = 1$}{
            $v.B \leftarrow \varnothing$\\   
        }\Else{        
            $v' \leftarrow$ 
            the Vote message that $u$ broadcast in round $r-1$\\   
            $v.B \leftarrow v'.B$\\     
        }  
        Broadcast $v$\\     	
    }
    
    \If{$u \text{ \upshape collects } 4f+1 \text{ \upshape Vote messages of round }
    r \text{ \upshape containing the 
    same value } b \text{ \upshape and } b \neq \varnothing$ $\text{ \upshape 
    before }TO_{commit} \text{ \upshape expires}$}{
        Commit $b$\\
    }\Else{
        Go to round $r+1$\\    
    }
\end{algorithm}

\noindent \textbf{The complete description of the first step:}
Let $lockset(r)$ be the set of Vote messages of round $r$ received 
by node $u$. $lockset(r)$ is \textbf{valid} if it contains
$4f+1$ Vote messages of round $r$.
We now describe the first step from the viewpoint of $u$ in detail.
Let $r_c$ be the current round.
If $u$ is not the leader of round $r_c$, 
then $u$ goes to the second step, i.e., voting.
Otherwise, if $u$ is the leader, $u$'s action depends on whether $r_c = 1$ or $r_c > 1$.

\noindent \textbf{Case 1 $(r_c = 1)$:}
$u$ constructs a Proposal message $p$, where $p.B$ is the initial value of $u$ 
and $p.R = 1$.

\noindent \textbf{Case 2 $(r_c > 1)$:}
In this case, if $u$ wants to send a Proposal message $p$, 
$u$ must have a valid lockset of the previous round, 
i.e., $|lockset(r_c-1)| \geq 4f+1$. 
Otherwise, if $|lockset(r_c-1)| < 4f+1$, then there is no Proposal message in round $r_c$.
For other nodes to verify this condition,
the Proposal message must contain $lockset(r_c-1)$.
We further impose a constraint on $p.B$.
If at least $2f+1$ votes in $lockset(r_c-1)$ contain the same candidate value $b$ 
and $b \neq \varnothing$, then $p.B=b$. 
Note that if multiple candidate values
satisfy the constraint, then the leader can choose any one of them. 
Otherwise, if no candidate value satisfies the constraint, 
$u$ can propose its own initial value. 

\section{Analysis}\label{sec:analysis}
\subsection{Proof of Safety}
To prove that our solution guarantees safety, 
it suffices to prove the following two claims.
We say a node votes for a value $b$ in round $r$ if the node sends a Vote message 
containing $b$ in round $r$.
\begin{claim}
If two non-faulty nodes $u_1$ and $u_2$ commit values $b_1$ and $b_2$ in the same round 
$r$, respectively, then $b_1 = b_2$. 
\end{claim}
\begin{proof}
For the sake of contradiction, assume that $b_1 \neq b_2$.
$u_1$ (respectively, $u_2$) receives $4f+1$ Vote messages 
containing $b_1$ (respectively, $b_2$).
Hence, in round $r$, 
at least $3f+1$ non-faulty nodes vote for $b_1$ and a different set of at least 
$3f+1$ non-faulty nodes vote for $b_2$. Thus, there are at least $6f+2 > n$ nodes, 
which is a contradiction.
\end{proof}

\begin{claim}
Once a non-faulty node commits value $b$ in round $r$, 
for any future round $r' > r$, only $b$ can be committed in round $r'$.
\end{claim}

\begin{proof}
Let $N_b$ be the set of nodes that vote for $b$ in round $r$.
Because $b$ is committed in round $r$, $|N_b| \geq 4f+1$.
Let $G_b$ be the subset of $N_b$ that contains non-faulty nodes only.
Thus, $|G_b| \geq 3f+1$.
Observe that if the leader of round $r+1$ has a valid lockset $LS$ of round $r$, 
then $LS$ must contain at least $2f+1$ Vote messages sent from $G_b$.
In addition, for each candidate value $b' \neq b$, 
at most $2f$ Vote messages in $LS$ contain $b'$. 
Thus, if the leader of round $r+1$ can send a Proposal message $p$, 
$p.B=b$ must hold.
Otherwise, if there is no Proposal message of round $r+1$, 
all nodes in $G_b$ still vote for the value that they vote for in round $r$, 
i.e., $b$. In both cases, all nodes in $G_b$ still vote for $b$ in round $r+1$.
The claim then follows by induction.
\end{proof}

\subsection{Proof of Liveness Under the Partially Synchronous Model}
A standard technique to guarantee liveness under the partially synchronous model
is to double the lengths of the timeouts (e.g., $TO_{vote}$ and $TO_{commit}$) 
whenever entering a new round~\cite{PBFT}.
It can be shown that there is some round $r$ such that for any round $r' > r$, 
all non-faulty nodes can receive messages from each 
other before the timeouts expire~\cite{PBFT}. 
Thus, in some round $r' > r+1$, the leader is non-faulty\footnote{Recall that the leader 
is chosen in a round-robin fashion. Hence, in some round $r' \in \{r+2, r+3, \cdots, r+f+2\}$, 
the leader is non-faulty.} and has a valid lockset of round $r'-1$. 
Thus, there must be a valid Proposal message $p$ in round $r'$. 
All $4f+1$ non-faulty nodes then vote for $p.B$ in round $r'$.
Since all these $4f+1$ Vote messages can be received in time, 
all non-faulty nodes can commit $p.B$ in round $r'$.

\bibliographystyle{plain}

\end{document}